\documentclass{article}

\usepackage{microtype}
\usepackage{graphicx}
\usepackage{subcaption}
\usepackage{booktabs} 

\usepackage{hyperref}

\usepackage[accepted]{icml2026}

\usepackage{amsmath}
\usepackage{amssymb}
\usepackage{mathtools}
\usepackage{amsthm}
\newcounter{remarkcounter}

\usepackage[capitalize,noabbrev]{cleveref}

\newcommand{\E}{\mathbb{E}}
\newcommand{\Prb}{\mathbb{P}}
\newcommand{\OPT}{\mathrm{OPT}}

\theoremstyle{plain}
\newtheorem{theorem}{Theorem}[section]
\newtheorem{proposition}[theorem]{Proposition}
\newtheorem{lemma}[theorem]{Lemma}
\newtheorem{corollary}[theorem]{Corollary}
\theoremstyle{definition}
\newtheorem{definition}[theorem]{Definition}

\theoremstyle{remark}
\newtheorem{remark}[remarkcounter]{Remark}
\usepackage[textsize=tiny]{todonotes}

\icmltitlerunning{Multi-Agent Combinatorial-Multi-Armed-Bandit framework for the Submodular Welfare Problem under Bandit Feedback}

\begin{document}

\twocolumn[
  \icmltitle{Multi-Agent Combinatorial-Multi-Armed-Bandit framework for the Submodular Welfare Problem under Bandit Feedback}



  



  \begin{icmlauthorlist}
    \icmlauthor{Subham Pokhriyal}{a}
    \icmlauthor{Shweta Jain}{a}
    \icmlauthor{Vaneet Aggarwal}{b}
  \end{icmlauthorlist}

  \icmlaffiliation{a}{Indian Institute of Technology Ropar, Rupnagar, India}
  \icmlaffiliation{b}{Purdue University, West Lafayette, USA}

  \icmlcorrespondingauthor{Subham Pokhriyal}{subham.22csz0002@iitrpr.ac.in}
  \icmlcorrespondingauthor{Shweta Jain}{shwetajain@iitrpr.ac.in}
  \icmlcorrespondingauthor{Vaneet Aggarwal}{vaneet@purdue.edu}

  \icmlkeywords{Machine Learning, ICML}

  \vskip 0.3in
]




\printAffiliationsAndNotice{} 


\begin{abstract}
 We study the \emph{Submodular Welfare Problem} (SWP), where items are partitioned among agents with monotone submodular utilities to maximize the total welfare under \emph{bandit feedback}. Classical SWP assumes full value-oracle access, achieving $(1-1/e)$ approximations via continuous-greedy algorithms. We extend this to a \emph{multi-agent combinatorial bandit} framework (\textsc{MA-CMAB}), where actions are partitions under full-bandit feedback with non-communicating agents. Unlike prior single-agent or separable multi-agent CMAB models, our setting couples agents through shared allocation constraints. We propose an explore--then--commit strategy with randomized assignments, achieving $\tilde{\mathcal{O}}(T^{2/3})$ regret against a $(1-1/e)$ benchmark—the first such guarantee for partition-based submodular welfare problem under bandit feedback.

\end{abstract}
\section{Introduction}

In combinatorial multi-armed bandit (CMAB) problems, a learner selects a subset of base arms at each round and receives stochastic feedback. The challenge stems from the structured yet exponentially large action space, which makes exploration and optimization nontrivial. In CMAB, two types of feedback are typically studied. With semi-bandit feedback, the learner observes the individual reward contribution of each chosen arm in the selected subset. This allows for efficient estimation and often leads to faster learning. However, in the full-bandit feedback setting, only the total reward from the chosen set is observed, without any information about individual arm contributions. This makes learning significantly harder, as marginal values must be inferred indirectly, often requiring deliberate and suboptimal exploration.

Despite its challenges, full-bandit feedback more accurately reflects many real-world scenarios where per-arm signals are inaccessible due to privacy concerns, measurement limits, or system constraints. Applications such as recommender systems, data summarization, fair allocation, revenue and influence maximization frequently operate in such regimes, where only aggregate outcomes are observable~\cite{fourati2023randomized,nie2023framework,fourati2024federated,pokhriyal2025anytime}. This motivates the need for learning algorithms that can operate effectively in the full-bandit setting.

Our work addresses a multi-agent variant of this problem, grounded in the well-studied Submodular Welfare Problem (SWP). In an offline SWP, the objective is to divide a set of items among multiple agents (partitions) so as to maximize the total (utilitarian) utility, where each agent's utility is defined by a monotone submodular function. This problem arises in domains such as economics, fair division, and resource allocation. In the value-oracle model, the problem admits a $1/2$-approximation via a greedy algorithm~\cite{fisher1978analysis,lehmann2001combinatorial}, while more advanced techniques (such as the continuous-greedy algorithm with pipage or randomized rounding) achieve the optimal $(1 - 1/e)$ approximation~\cite{vondrak2008optimal}. The SWP reduces to submodular maximization under a partition matroid.

Online submodular optimization under bandit feedback has gained considerable attention in recent literature, particularly in single-agent CMAB settings. For non-monotone objectives, \citet{fourati2023randomized} adapts the \emph{RANDOMIZED-SUM} algorithm from \citet{buchbinder2015tight}, establishing a sublinear $1/2$-regret. In the monotone case under cardinality constraints, \citet{nie2023framework} builds on the classical \emph{GREEDY} algorithm of \citet{nemhauser1978analysis}, while \citet{fourati2024federated} extends the \emph{STOCHASTIC-GREEDY} approach of \citet{mirzasoleiman2015lazier}, both of which yield a sublinear $(1 - 1/e)$-regret. 

However, all these works are limited to a single agent setting. Extensions to multi-agent scenarios have been proposed in federated and decentralized optimization contexts~\cite{konevcny2016federated,mcmahan2017communication,li2020federated,hosseinalipour2021federated,elgabli2022fednew,fourati2024filfl}, but these typically assume continuous updates and communication between agents. In contrast, we focus on discrete, non-communicating, multi-agent settings where global utility depends on structured item allocations, and feedback is limited to total reward.
Practical applications of this framework include partitioned recommendation, distributed sensing, influence maximization across user groups, and equitable allocation of indivisible goods.

In this work we formalize a connection between the offline SWP \cite{vondrak2008optimal} and the online multi-agent CMAB framework, with a particular focus on settings involving full-bandit feedback. In the online variant, the learner repeatedly selects a feasible partition of items among agents, receives aggregate reward feedback, and aims to minimize regret with respect to a benchmark offline solution.
\noindent \emph{Our formulation departs from prior work in four key aspects:} 
(i) all agents are present and active in every round (no partial participation), (ii) agents are non-communicating and for each agent $i$  we have distinct submodular utility  $w_i(\cdot)$, (iii) the learner selects a \emph{partition} of the item set, assigning disjoint subsets to agents, rather than a single best subset, and (iv) the optimization objective is the total welfare $\sum_i w_i(\cdot)$, coupling feasibility across agents. 



To the best of our knowledge, this is the first work to cast \emph{discrete partitioning for submodular welfare} as a multi-agent CMAB problem under full-bandit feedback. We design a learning algorithm that lifts offline approximation guarantees to the online setting and achieves a regret bound of $\tilde{\mathcal{O}}(T^{2/3})$ using discrete randomized allocation, adaptive exploration, and full-bandit feedback. We outline the principal contributions of this work as follows:
\begin{itemize}
  \item We introduce a non-communicating multi-agent partitioning framework, where each agent receives a distinct (possibly monotone) submodular reward and only aggregate outcomes are observed.
  \item We establish new \emph{resilience guarantees} for the offline submodular welfare problem under noisy value oracles, applicable to partition/matroid/knapsack-style constraints.
  \item We propose an offline-to-online reduction that achieves $\tilde{\mathcal{O}}(T^{2/3})$ regret with respect to an $\alpha$-approximate benchmark under standard noise models (bounded or sub-Gaussian).
  \item We design a discrete, randomized assignment policy that approximates the continuous-greedy trajectory via adaptive sampling and deliberate exploration based only on total rewards.
\end{itemize}

\section{Related Work}

\begin{table*}[t]
\centering
\caption{Summary of Resilience Parameters and Regret Guarantee for MA-CMAB. This table outlines the key components underlying the regret analysis of the MA-CMAB algorithm for online submodular welfare maximization. It includes the approximation factor $\alpha = 1 - 1/e$, the robustness parameter $\delta$ capturing noise sensitivity, and the oracle complexity $\eta$. The framework achieves a sublinear $\mathcal{O}(T^{2/3})$ regret bound with respect to the best offline allocation, under stochastic full-bandit feedback.}
\label{tab:summaryofresults}
\begin{tabular}{c c c | c c c}
\hline
\textbf{App.} & \textbf{Objective $f$} & \boldmath$\alpha$ & \boldmath$\delta$ & \boldmath$\eta$ & \textbf{Regret Bound} \\
\hline
USW & $\max\texttt{(Sub)}$ & $1 - \frac{1}{e}$ 
& $(4MN + 2M)$ 
& $(MN)^8$ 
& $\mathcal{O}\!\left(\delta^{2/3} \eta^{1/3} C^{2/3} T^{2/3} \log^{1/3}(T)\right)$ \\
\hline
\end{tabular}
\end{table*}

In this section, we review related work relevant to the construction of an online multi-agent CMAB framework for the offline Submodular Welfare Problem (SWP), organized into three key areas:
(1) Submodular Welfare Optimization – studies item allocation for maximizing welfare with approximation guarantees under value and demand oracle models.
(2) Multi-agent MAB (MA-MAB) – explores decentralized learning with limited communication and agent-specific rewards.
(3) Combinatorial Bandit Frameworks – adapts offline submodular algorithms to online settings, recently extended to multi-agent cases.
These lines of work collectively motivate our decentralized, non-communicating, online approach to SWP.

\subsection{Submodular Welfare Problem}

The Submodular Welfare Problem (SWP) arises naturally in combinatorial auctions, where multiple agents compete for a set of items and derive value from combinations of them. Each agent is associated with a private utility function that evaluates the value of any subset of items. The objective is to allocate items among agents in a way that maximizes the overall (total) welfare. Since the input space is exponential in size, algorithmic access is typically assumed to be restricted to oracle (query) access to the valuation functions.

Two standard oracles are commonly used. The \emph{value oracle model} allows the algorithm to evaluate the utility of particular item sets via direct queries. In the more expressive \emph{demand oracle model}, given item prices, the oracle returns the player’s most preferred bundle under those prices. These models allow optimization over general submodular functions without assuming explicit representation.

In the value oracle setting, a greedy algorithm provides a $1/2$-approximation~\cite{lehmann2001combinatorial}, while ~\cite{vondrak2008optimal} later achieved an optimal approximation ratio of $(1-1/e)$ using the continuous greedy technique and pipage rounding. In the demand oracle model, the same $(1-1/e+\varepsilon')$ approximation factor can be achieved using oblivious rounding procedures~\cite{feige2010submodular}. However, improving beyond this threshold is NP-hard~\cite{khot2005inapproximability}, and even matching it with only value queries requires exponentially many calls~\cite{mirrokni2008tight}.


In particular, SWP is a special case of maximizing a monotone submodular function subject to matroid constraints, first studied by~\cite{fisher1978analysis}, with reduction formalized in~\cite{lehmann2001combinatorial}. In the uniform matroid case, where only $k$ items can be selected, the greedy algorithm achieves $(1 - 1/e)$~\cite{nemhauser1978analysis}. For general matroids, it provides only a $1/2$-approximation. A major breakthrough was achieved by \citet{calinescu2007maximizing}, who proposed the continuous greedy algorithm. This method operates on a fractional relaxation of the submodular function’s multilinear extension and subsequently applies pipage rounding~\cite{ageev2004pipage} to recover a feasible integral solution without compromising the approximation guarantee. This method achieves the tight $(1 - 1/e)$ bound in the value oracle model for general matroids. In this work, we adopt continuous greedy value-oracle benchmark  of~\cite{vondrak2008optimal} as the offline reference for our stochastic variant of the SWP in a Multi-Agent Combinatorial Multi-Arm Bandit (MA-CMAB) framework.

\subsection{Multi-agent Multi-armed Bandits}
Motivated by the single-agent bandit setting \cite{pmlr-v108-chawla20a}, ~\cite{chawla2023collaborative} proposed decentralized algorithms that facilitate agent collaboration through selective information sharing. Recent works in multi-agent multi-armed bandits (MA-MAB) consider various communication-based models, including broadcasting \cite{agarwal2022multi,buccapatnam2015information,chakraborty2017coordinated}, pairwise gossip \cite{pmlr-v108-chawla20a,sankararaman2019social}, and decentralized settings with frequent interactions \cite{kolla2018collaborative,lalitha2021bayesian,martinez2019decentralized}. Contextual and linear-reward MA-MAB variants have also been explored \cite{dubey2020kernel,tekin2015distributed,dubey2020differentially,chawla2022multiagent,korda2016distributed}. A few works target minimizing simple regret instead of cumulative regret \cite{hillel2013distributed,szorenyi2013gossip}, and others allow heterogeneous bandits across agents \cite{bistritz2018distributed,shahrampour2017multi,reda2022near,zhu2021federated}. Additionally, collision-based models have received attention \cite{anandkumar2011distributed,avner2014concurrent,bistritz2018distributed,boursier2019sic,dakdouk2021collaborative,kalathil2014decentralized,liu2010distributed,liu2020competing,mansour2017competing,rosenski2016multi}.

Our setting differs from collision-based MA-MAB models, as there are no collisions among arms, each agent operates an independent bandit with its own reward distribution. Consequently, the reward obtained by different agents for pulling the same arm may differ, which fundamentally distinguishes our framework from traditional collision-based settings.
However, agents still compete with each other as the allocation of items comprises a partition. 

\subsection{Single/Multi-agent CMAB Framework}
Combinatorial bandit frameworks in the single-agent setting have been extensively studied. ~\citet{niazadeh2021online} introduced a general framework for adapting iterative greedy algorithms from offline settings to adversarial combinatorial bandit problems, under the requirement that the offline algorithm exhibits an iterative greedy structure and a property known as Blackwell reproducibility. More recently, ~\cite{nie2023framework} introduced the C-ETC framework, which extends offline algorithms with robustness guarantees to the stochastic CMAB setting. This framework generalizes earlier approaches for submodular maximization~\cite{nie2022explore,fourati2023randomized}, though it does not subsume more recent work such as ~\cite{fourati2024combinatorial}, which offers refined guarantees for submodular bandits. These works collectively form the foundation of single-agent combinatorial bandit theory, particularly in scenarios involving robustness and offline-to-online reductions.

While the C-ETC framework~\cite{nie2023framework} provides an offline-to-online reduction for stochastic CMABs in the single-agent setting, it does not extend to recent advances in multi-agent learning. In contrast, ~\cite{fourati2024federated} propose a generalized framework that accommodates multiple communicating agents ($M \geq 1$), allowing coordination across disjoint combinatorial actions. This yields tighter regret guarantees by reducing exploration variance through distributed sampling and shared information.

Our setting differs from both single-agent and prior multi-agent CMAB frameworks in several key aspects. Unlike single-agent models, we consider multiple agents acting concurrently, each with their own distinct submodular utility function. All agents participate in every round, but they are non-communicating, making coordination more challenging. Unlike recent multi-agent approaches such as ~\cite{fourati2024federated}, which focus on selecting a single best subset shared among agents, our model requires computing a disjoint partition of items, i.e., one subset per agent with the goal of maximizing the total submodular welfare $\sum_i w_i(S_i)$. Building on the single-agent foundation of ~\cite{nie2023framework}, we extend to this decentralized multi-agent partitioning setting and provide online guarantees with sublinear regret subject to matroid constraints.

\section{Problem Statement}\label{sec:prob_state}
\paragraph{Model.}
We consider a fair division setting with $M$ agents $P = \{1, \dots, M\}$ and a finite collection of $N$ indivisible items, denoted by $Q$. The ground set of all possible agent-item assignments is defined as $X = P \times Q$, where an element $(i,j) \in X$ represents the assignment of item $j$ to agent $i$. A feasible allocation is a subset $S \subseteq X$ that satisfies allocation constraints specified by a \textit{partition matroid}, with optional per-agent quotas:
\[
\mathcal{D}
\;=\;
\Bigl\{\,S\subseteq X:\;
\forall j\in Q,\ |S\cap (P\times\{j\})|\le 1
\ \ \text{and}\ \ \]
\[
\forall i\in P,\ |S\cap(\{i\}\times Q)|\le b_i
\Bigr\},
\]
where $b_i \in \mathbb{Z}_{\ge 0}$ is the (optional) quota of agent $i$. Setting $b_i = N$ for all $i$ removes agent-side capacity constraints. Each feasible allocation $S \in \mathcal{D}$ induces a bundle $s_i = \{j \in Q : (i,j) \in S\}$ for each agent $i$, which together form a (possibly capped) partition of $Q$. Each agent $i$ is equipped with a \textit{monotone submodular valuation function} $w_i : 2^Q \to [0,1]$, with $w_i(\emptyset) = 0$. An allocation 
$S $ has (utilitarian) welfare defined as the total sum of the agents valuations:
\(
f(S)\;=\;\sum_{i=1}^{M} w_i(s_i)\;\in\; [0,1].
\)

\texttt{Example.}
The proposed model subsumes stochastic combinatorial auctions with noisy
valuation access; see Appendix~\ref{app:auction} for a detailed instantiation.
\paragraph{Offline Benchmark.}
Let $\OPT \in \arg\max_{A \in \mathcal{D}} f(A)$ denote an optimal allocation maximizing total welfare. Since exact maximization over $\mathcal{D}$ is NP-hard in general, we assume access to an $\alpha$-approximate polynomial-time offline algorithm $\mathcal{A}{lg}_{off}$ that returns a feasible solution $A_{\mathcal{A}{lg}} \in \mathcal{D}$ satisfying:
\[
f(A_{\mathcal{A}{lg}}) \;\ge\; \alpha \cdot f(\OPT), \quad \text{for some } \alpha \in (0,1].
\]
This offline guarantee serves as the benchmark for our online learning protocol.

\paragraph{Online Learning Protocol.}
Time proceeds in discrete rounds indexed by $t = 1, \dots, T$. At each round $t$, a learning policy selects a feasible allocation $A_t \in \mathcal{D}$, and receives a stochastic full-bandit reward:
\[
r_t = f_t(A_t),\quad\text{where } f_t(A_t) \in [0,1] \text{ and } \mathbb{E}[f_t(A_t)] = f(A_t).
\]
The reward sequence $\{f_t(A_t)\}_{t \ge 1}$ is assumed to be i.i.d. across rounds for each fixed allocation $A_{\mathcal{A}{lg}} \in \mathcal{D}$, and only the aggregate reward $f_t(A_t)$ is observed after each round— no agent or item level
feedback is available.
\paragraph{Regret for Submodular Welfare Problem.}
The goal is to minimize regret against the offline benchmark. Let $\mathcal{R}{(T)}$ denote the cumulative $\alpha$-regret after $T$ rounds, defined as:
\begin{equation}
\label{eq:alpha-regret-partition}
\mathbb{E}[\mathcal{R}{(T)}]
\;=\;
\alpha\,T\,f(\OPT)
\;-\;
\mathbb{E}\!\left[\sum_{t=1}^{T} f_t(A_t)\right].
\end{equation}
This is a specialization of the standard CMAB pseudo-regret to our setting of partitioned actions and submodular social welfare. A smaller $\mathbb{E}[\mathcal{R}{(T})]$ indicates better learning relative to the best achievable offline approximation.
\section{Resilience Guarantee for the Submodular Welfare Problem}

\paragraph{Continuous Greedy for Submodular Welfare.}
\citet{wolsey1982analysis} early work on real-valued submodular maximization over knapsack polytopes updates one coordinate at a time and attains only a $1/2$-approximation in that setting. In contrast, the \emph{continuous greedy} method of ~\cite{vondrak2008optimal} moves in a globally chosen ascent direction inside the matroid polytope and achieves a $(1-1/e)$-approximation; a discrete variant with pipage rounding yields integral solutions with the same expectation.

Let $P=\{1,\dots,M\}$ be the set of agents and $Q=\{1,\dots,N\}$ the set of items. Each agent $i\in P$ has a monotone submodular valuation $w_i:2^Q\to\mathbb{R}_+$. Define the ground set $X=P\times Q$ and, for any $S\subseteq X$, let $s_i=\{\,j\in Q:(i,j)\in S\,\}$ be agent $i$'s bundle. For $y\in[0,1]^X$, define a random vector $\hat y\in\{0,1\}^X$ by rounding each coordinate independently: set $\hat y_j=1$ with probability $y_j$ and $\hat y_j=0$ with probability $1-y_j$.
The welfare function is
\(
f(S)=\sum_{i=1}^M w_i(s_i),
\ \max_{S\in\mathcal{I}} f(S),
\)
subject to the partition matroid
\(
\mathcal{I}=\bigl\{\,S\subseteq X:\ \forall j\in Q,\ |\{\,i\in P:(i,j)\in S\,\}|\le 1\,\bigr\}.
\)
Its polytope is
\(
\mathcal{P}(\mathcal{M})=\Bigl\{\,y\in[0,1]^{M\times N}:\ \sum_{i=1}^M y_{ij}\le 1\ \ \forall j\in Q\Bigr\}.
\)
\begin{center}
\begin{algorithm}[t]
\caption{\textsc{Continuous Greedy (Partition Matroid)}}
\label{alg:continuous-greedy-boxed}
\begin{algorithmic}[1]
\REQUIRE \small Stepsize $\lambda \in (0,1]$ (e.g., $\lambda = \tfrac{1}{(MN)^2}$)
\STATE Initialize $t \gets 0$, and set $y_{ij}(0) = 0$ for all $i \in [M], j \in [N]$
\WHILE{$t < 1$}
    \FOR{each $i \in [M], j \in [N]$}
        \STATE Estimate marginal: 
        \[
        \omega_{ij}(t) \approx \mathbb{E}[w_i(R_i(t) \cup \{j\}) - w_i(R_i(t))]
        \]
        where $R_i(t)$ includes each $j' \in Q$ independently with probability $y_{ij'}(t)$
    \ENDFOR
    \FOR{each item $j$}
        \STATE Pick $i_j(t) \in \arg\max_i\ \omega_{ij}(t)$ and update:
        \[
        y_{i_j(t),j}(t+\lambda) \gets y_{i_j(t),j}(t) + \lambda
        \]
        \[
        y_{ij}(t+\lambda) \gets y_{ij}(t),\quad \forall i \ne i_j(t)
        \]
    \ENDFOR
    \STATE $t \gets t + \lambda$
\ENDWHILE
\STATE \textbf{Rounding:} Use pipage rounding to convert $y(1)$ to integral $S \in \mathcal{I}$, preserving $\mathbb{E}[f(S)] = F(y(1))$
\end{algorithmic}
\end{algorithm}
\end{center}
\allowdisplaybreaks
The multilinear extension of $f$ is
\begin{align}
\allowdisplaybreaks
F(y)
&:= \mathbb{E}\bigl[f(\widehat{y})\bigr] \notag 
= \sum_{R \subseteq X} f(R)
   \prod_{i \in R} y_i
   \prod_{j \notin R} (1-y_j) \notag \\
&= \sum_{i=1}^M \mathbb{E}\bigl[w_i(R_i)\bigr].
\end{align}
where for each agent $i$, $R_i\subseteq Q$ is formed by including each item $j$ independently with probability $y_{ij}$, and the collections $\{R_i\}_{i=1}^M$ are mutually independent. Equivalently, item $j$ is assigned to agent $i$ with probability $y_{ij}$ (and unassigned with probability $1-\sum_i y_{ij}$), which respects the partition constraint in expectation.

\noindent\textbf{Approximation Guarantees.}
By the analysis of~\cite{vondrak2008optimal}, the above \cref{alg:continuous-greedy-boxed} returns (after rounding) an allocation $S$ satisfying,
\(
\mathbb{E}[f(S)]\ \ge\ (1-1/e)\cdot f(\OPT),
\)
 where \(\OPT\) denotes the optimal feasible allocation.
\paragraph{Oracle complexity.}
Let $S$ denote the number of value-oracle samples used to estimate each $\omega_{ij}(t)$ at a time step. With step size $\lambda$, the number of iterations is $1/\lambda$, and each iteration touches $MN$ pairs, so the total value-oracle calls are
\[
\underbrace{\tfrac{1}{\lambda}}_{\text{iterations}}
\times
\underbrace{MN}_{\text{pairs}}
\times
\underbrace{Z}_{\text{samples per pair}}
\ =\ O\!\left(\tfrac{M\,N\,Z}{\lambda}\right).
\]

For the canonical discretization $\lambda=\tfrac{1}{(MN)^2}$ and a conservative choice $Z=(MN)^5$, this upper-bounds to $(MN)^8$ calls, matching the standard worst-case accounting used to secure the approximation guarantee in the value-oracle model.

\subsection{Offline Resilient Approximation}
In the Submodular Welfare Problem, where the objective is to maximize \( f(S) = \sum_i f(s_i) \), with \( S = \{s_1, s_2, \ldots, s_M\} \) being a disjoint partition of items among agents, the Continuous Greedy algorithm achieves the following guarantee:

\begin{definition}[$(\alpha, \delta, \eta)$- Resilient Approximation]\label{def:robustness}  
An offline algorithm \(\mathcal{A}{lg}\) is an \((\alpha, \delta, \eta)\)-resilient 
if, given an access to approximate oracle 
$\hat{f}$
satisfying \(|f(S) - \hat{f}(S)| < \epsilon\) for all \(S \in \Omega\), 
\(\mathcal{A}{lg}\) returns a solution \(S_{\mathcal{A}{lg}}\) such that:  
\begin{align}
    \mathbb{E}[f(S_{\mathcal{A}{lg}})] \geq \alpha f(\OPT) - \delta \epsilon,
\label{eq:def:resilience}
\end{align}  
where \(f(\OPT) = \arg\max_{S \subseteq X} f(S)=  \arg\max_{S \subseteq X}\sum_{i \in M}w_i(s_i)\). 

Here, \(\eta\) bounds the total number of oracle calls to \(\hat{f}\), and \(\delta\) quantifies resilience to approximation errors.  
In our setting, the oracle complexity parameter $\eta$ from
\cref{def:robustness} coincides with the number of distinct actions evaluated during
the exploration phase; we therefore use a single symbol $\eta$ for both
quantities.
\end{definition} 
\begin{remark}
The robustness notion of~\cite{niazadeh2021online} is tied to approximation algorithms with an \emph{iterative greedy} structure. In contrast, Definition~\ref{def:robustness} imposes no such structural assumptions.
\end{remark}

\begin{remark}
The noisy oracle \(\hat{w}_i(s_i)\) need not preserve monotonicity or submodularity, which is explicitly accounted for in our resilience analysis.
\end{remark}
\begin{remark}
When $f$ and $\hat{f}$ take values in $[0,1]$, a multiplicative error bound
$\lvert f(S)-\hat f(S)\rvert \le \varepsilon' f(S)$
implies the additive bound
$\lvert f(S)-\hat f(S)\rvert \le \varepsilon'$.
\end{remark}

\subsection{Resilience Guarantees for Continuous Greedy Algorithm under Noisy Oracle Access}

Resilience means that even if utility evaluations are noisy within 
$\epsilon$, the algorithm still achieves near-optimal performance up to an additive error proportional to $\epsilon$ i.e. $\delta \epsilon$, without breaking the 
$1-\frac{1}{e}$-approximation guarantee.
We now show that the \textsc{Continuous Greedy} algorithm satisfies
\cref{thm:CG:robust} under additive oracle noise.

\begin{theorem}[\textsc{Continuous Greedy} Resilience]
\label{thm:CG:robust}
Under inexact utility evaluations \(|\hat{w}(S) - w(S)| \leq \epsilon\), for $\epsilon \le \frac{1}{(MN)^2}$, \textsc{continuous greedy} is an \((\alpha, \delta, \eta)\)-resilient approximation algorithm for the Submodular Welfare (discrete partition) problem, where:
\[
\alpha = 1 - \frac{1}{e}, \quad 
\delta = (4MN + 2M), \quad
\eta = (MN)^8.
\]
\end{theorem}

\noindent
The key resilience parameters \(\alpha, \delta\ and,\ \eta\) are summarized in ~\cref{tab:summaryofresults}.
We next prove the \cref{thm:CG:robust} by series of Lemmas.
\cref{lemma:PR} is a standard result and is used as a black-box rounding procedure throughout the paper.

For $y \in [0,1]^X$ and random vector $\hat{y} \in \{0,1\}^X$ independently rounded by  $0/1$ corresponds to probability $y_j$.

\begin{lemma}\label{lem:cg-upper-bound}
Let \( y \in [0,1]^X \) be a fractional solution and let \( R \) denote a random set corresponding to \(\hat{y}\), obtained by sampling each element \( j \in X \) independently with probability \( y_j \).  
Let \(\OPT = \max_{S \in \mathcal{I}} f(S)\).  
If all oracle evaluations are noisy with additive error at most \(\varepsilon\), i.e.
\[
|\hat{f}(S) - f(S)| \leq \varepsilon \quad \forall S \subseteq X,
\]
then, the true optimum satisfies the following upper bound:
\[
f(\OPT) \le F(y) + \max_{I \in \mathcal{I}} \sum_{j \in I} \mathbb{E}[f_R(j)] + (4MN + 2M)\varepsilon,
\]
where:
\begin{itemize}
    \item $F(y)$ is the true multilinear extension of $f$ evaluated at $y$,
    \item $f_R(j) := f(R \cup \{j\}) - f(R)$ is the true marginal contribution of $j$ under random set $R$,
    \item and $\mathcal{I}$ is the family of feasible integral solutions (e.g., a partition matroid).
\end{itemize}
\end{lemma}

\begin{proof}[Proof sketch]

Fix an optimal solution $\OPT \in\mathcal{I}$. By submodularity,
for any random set $R$, we have
\[
f(\OPT) \le f(R) + \sum_{j\in \OPT} \big(f(R\cup\{j\}) - f(R)\big).
\]
Taking expectation and carefully accounting for oracle noise yields
the stated bound. The complete argument is provided in
\cref{app:lem:cg-upper-bound} (~\cref{app:cg-resilience}).
\end{proof}
In the continuous greedy algorithm, at each time step \(t\), we want to know how valuable each 
element \(j\) is if we add it next. This value is the \emph{expected marginal gain}:
\[
\omega_j(t) \;=\; \mathbb{E}\big[f(R(t)\cup \{j\}) - f(R(t))\big],
\]
where \(R(t)\) is a random set formed by including each element independently with probability 
\(y_i(t)\). Since this expectation cannot be computed directly, the algorithm estimates it 
by \emph{Monte Carlo sampling}: draw about \((MN)^5\) random sets \(R(t)\), query the oracle for each 
sample, compute the marginal contribution of \(j\), and average. This gives an accurate estimate 
\(\hat{\omega}_j(t)\) with high probability.

Once we have \(\omega_j(t)\) for every element, we select a \emph{maximum-weight independent set} 
\(I(t)\) under the matroid constraints, using \(\omega_j(t)\) as weights. Intuitively, \(I(t)\) is the 
best feasible subset to move towards, and we update the fractional vector in that direction:
\[
y(t+\lambda) \;=\; y(t) + \lambda \cdot 1_{I(t)}.
\]
Iterating this process makes \(y(1)\) a convex combination of independent sets, which can then be 
rounded to an integral solution.

\begin{lemma}[Noisy resilient bound for Continuous Greedy]
\label{lem:noisy-progress}
Let $y$ be the fractional solution returned by
\textsc{Continuous Greedy} under noisy oracle access.
Then
\[
F(y)
\ge
\left(1-\frac{1}{e}\right)f(\OPT)-(4MN+2M)\varepsilon.
\]
\end{lemma}

\begin{proof}[Proof sketch]
The argument combines the noisy upper bound on the optimal value
(Lemma~\ref{lem:cg-upper-bound}) with a stepwise analysis of the continuous
greedy update. A full derivation of the recurrence is provided in
Appendix~\ref{app:cg-resilience}, see
Lemmas~\ref{app:lem:cg-upper-bound} and~\ref{app:lem:noisy-progress}.
\end{proof}


\begin{remark}  
In the noiseless case ($\varepsilon=0$), the additive error vanishes and we recover the classical
$(1-1/e)$-approximation guarantee.  
With noisy oracles, the additive error grows only linearly in $(MN)\varepsilon$, proving resilience. 
\end{remark}

\paragraph{Pipage Rounding}To establish resilience, our approach leverages a powerful rounding technique first introduced by \cite{ageev2004pipage}, and subsequently refined for matroid polytopes by \cite{calinescu2007maximizing}. We incorporate this method as a modular, black-box tool within our framework.

\begin{lemma}\label{lemma:PR}
There exists a randomized algorithm that runs in polynomial time and, given oracle access to a matroid \( \mathcal{M} = (X, \mathcal{I}) \), a monotone submodular function \( f: 2^X \rightarrow \mathbb{R}_+ \) (accessible via a value oracle), and a fractional vector \( y \in P(\mathcal{M}) \), returns an independent set \( S \in \mathcal{I} \) such that, with high probability,
\[
f(S) \geq (1 - o(1)) \cdot \mathbb{E}[f(\hat{y})],
\]
where \( \hat{y} \subseteq X \) is a random subset in which each element \( j \in X \) is included independently with probability \( y_j \). The \( o(1) \) term can be made polynomially small in \( M = |X| \).
\end{lemma}

This result implies that any fractional solution \( y \in P(\mathcal{M}) \) can be efficiently converted into an integral solution with only a negligible loss in expected function value. Therefore, it suffices to consider the continuous optimization problem
\[
\max \left\{ F(y) : y \in P(\mathcal{M}) \right\},
\]
where \( F \) denotes the multilinear extension of \( f \), instead of directly optimizing the discrete problem \( \max \left\{ f(S) : S \in \mathcal{I} \right\} \).
\paragraph{Resilient Approximation via Pipage Rounding.}
Let $y \in P(\mathcal{M})$ be the fractional solution obtained by the Continuous Greedy Algorithm with noisy marginal estimates bounded by $\varepsilon$. Then, the multilinear extension satisfies
\[
F(y) \ge \left(1 - \frac{1}{e}\right) f(\OPT) - (4MN + 2M)\varepsilon .
\]
Applying pipage rounding (Lemma \ref{lemma:PR}) to $y$ produces an integral independent set $S \in \mathcal{I}$ such that
\begin{align}
&f(S) \ge (1 - o(1)) \, F(y) \notag \\
&\ge \left(1 - \frac{1}{e} - o(1)\right) f(\OPT) - (4MN + 2M)\, \varepsilon.
\end{align}
where $o(1)$ is the discretization error of the continuous greedy process and can be made arbitrarily small as $M \to \infty$.

\begin{remark} For noiseless oracles, the additive $\varepsilon$ term vanishes, recovering the classical $1 - 1/e - o(1)$ approximation guarantee.
\end{remark}

\section{
Algorithmic Framework and Analysis
}
\subsection{MA-CMAB for Social Welfare: Offline to Stochastic}
\begin{center}
    \begin{minipage}{0.45\textwidth}
    \vspace{-.21in}
\begin{algorithm}[H]
\caption{\textsc{Submodular welfare MA-CMAB Algorithm}}\label{alg:bi-criteria-algo}
\begin{algorithmic}[1]
\REQUIRE \small Horizon $T$, ground set $X$, 
$(\alpha,\delta,\eta)$-resilient algorithm $\mathcal{A}{lg}$.
\STATE Set $m \gets \left\lceil \frac{\delta^{2/3} T^{2/3}M^{2/3} (\log T)^{1/3} }{2{\eta}^{2/3}}\right\rceil$
\STATE \textbf{Exploration Phase:}
\WHILE{$\mathcal{A}{lg}$ queries action $S$($m$-times)}
    \FOR{$i=1$ to $M$ agents}
        \STATE Select set $S$, observe $w_{ij}(s_i)'s\ \ \forall i$'s
    \ENDFOR
    \STATE Compute $\bar{w_i}(s_i) = \frac{1}{m}\sum_{k=1}^m w_{i,j}(s_i) \ \ \forall i$'s

    \STATE Compute $\bar{f}(S) = \sum_{i=1}^M w_i(s_i)$
     
    \STATE Return $\bar{f}(S)$ to $\mathcal{A}{lg}$
\ENDWHILE
\STATE \textbf{Exploitation Phase:}
\STATE Let $S \gets$ output of $\mathcal{A}$
\WHILE{$t\leq T$}
    \STATE Play $S$
\ENDWHILE
\end{algorithmic}
\end{algorithm}
 \vspace{-.21in}
 \end{minipage}
  \end{center}
\newcommand{\mopt}{m^*}
\newcommand{\moptrnd}{m^\xi}

The proposed \textsc{MA-CMAB} algorithm \ref{alg:bi-criteria-algo}, addresses the complex problem of \emph{multi-agent allocation under submodular social welfare} in a stochastic setting through a principled offline-to-online reduction. The algorithm is divided into two phases: an \textbf{exploration phase}, in which empirical estimates are collected over $m = \left\lceil \frac{\delta^{2/3} T^{2/3} M^{2/3} (\log T)^{1/3} }{2{\eta}^{2/3}} \right\rceil$ rounds, and an \textbf{exploitation phase}, where the algorithm commits to a partition $S = \{s_1, \ldots, s_M\}$ returned by a $(\alpha, \delta, \eta)$-resilient offline oracle $\mathcal{A}{lg}$. During exploration, each queried action $S$ is evaluated across agents by observing stochastic rewards $w_{ij}(S_i)$, and the corresponding estimates $\bar{w}_i(s_i)$ and $\bar{f}(S)$ are averaged and passed back to $\mathcal{A}{lg}$. This mechanism enables robust selection of a high-welfare allocation using offline optimization over empirical means. The design draws inspiration from the \texttt{ETC} procedure for non-combinatorial bandits \cite{slivkins2019introduction}, extending it to a multi-agent, submodular, and combinatorial domain.

The regret analysis is structured around a high-probability ``clean event'' $\mathcal{E}$, where empirical estimates concentrate around their true means (as shown in \cref{lem:clean}), and a low-probability ``unclean event,'' whose regret contribution is bounded separately. Under the clean event, the output of $\mathcal{A}{lg}$ yields an $(\alpha, \delta)$-approximate allocation, ensuring near-optimal Utilitarian Social Welfare (USW) during the exploitation phase. Our setting is more challenging than prior work such as \texttt{C-ETC} \cite{nie2023framework}, due to the need to (i) select disjoint subsets of items for multiple agents, (ii) optimize a monotone submodular utility function under partitioning constraints, and (iii) operate without simple marginal gain feedback. The resulting allocation problem is NP-hard even offline, as it generalizes the discrete submodular partitioning problem. Consequently, our framework must address the more challenging task of identifying optimal agent-item allocations based  on stochastic reward feedback. This positions the MA-CMAB algorithm as a significant advancement for online social welfare maximization in structured, uncertain, and combinatorial multi-agent environments.
\section{Offline to Stochastic Regret}

\subsection{Notation Overview}

We consider a horizon of \(T\) rounds, \(M\) agents, and \(N\) items. A feasible \emph{action} \(A\) is an allocation obeying the problem constraints (e.g., a partition of items to agents). At round \(t\), playing \(A\) yields a bounded stochastic reward \(f_t(A)\in[0,1]\) with mean \(f(A):=\E[f_t(A)]\). During exploration, the algorithm selects \(\eta\in\mathbb{N}\) actions \(S_1,\ldots,S_\eta\) and plays each action \(m\in\mathbb{N}\) times; the empirical mean is \(\bar f(S_i):=\frac{1}{m}\sum_{\ell=1}^m f_{t_{i,\ell}}(S_i)\). Let \(T_i\) be the time index by which action \(S_i\) has been completed \(m\) plays, and set \(T_0=0\), \(T_{\eta+1}\leq T\). 
We use the confidence radius
\begin{equation}\label{eq:rad}
 \mathrm{rad}\ :=\ \sqrt{\frac{\log T}{2m}}.
\end{equation}
The offline subroutine is \((\alpha,\delta)\)-robust in the sense that if the inputs are uniformly within \(\mathrm{\epsilon}\) of the truth on all actions it evaluates, then the returned allocation \(S\) satisfies
\begin{equation}\label{eq:robust}
  \E[f(S)]\ \ge\ \alpha\,f(\OPT)\ -\ \delta\,\mathrm{\epsilon}.
\end{equation}
If the final score aggregates \(M\) agentwise terms via a \(1\)-Lipschitz map, we set \(C:=M\) and \(\mathrm{\epsilon}=C\cdot\mathrm{rad}\); otherwise \(C:=1\) and \(\mathrm{\epsilon}=\mathrm{rad}\). The \(\alpha\)-regret is
\begin{equation}\label{eq:regret-def}
  \mathcal{R}(T)\ :=\ \sum_{t=1}^T \big(\alpha\,f(\OPT) - f_t(S_t)\big).
\end{equation}
\subsection{Clean Event and Concentration}
We define a high-probability event under which all empirical estimates used
by the offline subroutine are uniformly accurate.

\begin{definition}[Clean event]\label{def:clean}
Define
\begin{equation}\label{eq:clean-event}
  \mathcal{E}\ :=\ \bigcap_{j=1}^{\eta}\left\{\,\big|\bar f(S_j)-f(S_j)\big|<\mathrm{rad}\right\}.
\end{equation}
In the aggregated model use the threshold \(C\,\mathrm{rad}\) with the same \(\mathrm{rad}\) as in \eqref{eq:rad}.
\end{definition}

\begin{lemma}[Probability of the clean event]\label{lem:clean}
With \(\mathrm{rad}\) as in \eqref{eq:rad}, we have
\begin{equation}\label{eq:clean-prob}
  \Prb(\mathcal{E})\ \ge\ 1-\frac{2\eta}{T}.
\end{equation}
\end{lemma}
\begin{proof}[Proof sketch]
Lemma~\ref{lem:clean} follows from Hoeffding’s inequality and a union bound
over explored actions. A complete proof is provided in
Appendix~\ref{app:clean}.
\end{proof}

\subsection{Robust Approximation Transfer Under Clean Estimates:}
Under the clean event, uniform concentration ensures that the offline oracle’s approximation guarantee transfers to the true objective up to an additive error.
\begin{lemma}[Robust transfer]\label{lem:robust-transfer}
Let \(\tilde f\) be the empirical objective and suppose \(\|\tilde f-f\|_\infty\le \mathrm{\epsilon}\) on all actions the subroutine evaluates. If the subroutine is \((\alpha,\delta,\eta)\)-robust as in \eqref{eq:robust}, then for its output \(S\),
\begin{equation}\label{eq:transfer}
  \E[f(S)]\ \ge\ \alpha f(\OPT)\ -\ \delta\,\mathrm{\epsilon}.
\end{equation}
In particular, on \(\mathcal{E}\) with \(\mathrm{\epsilon}=C\,\mathrm{rad}\),
\begin{equation}\label{eq:transfer-gap}
  \alpha f(\OPT)-\E[f(S)]\ \le\ \delta\,C\,\mathrm{rad}.
\end{equation}
\end{lemma}

\begin{proof}[Proof sketch]
The result follows directly from the $(\alpha,\delta)$-robustness
definition applied to the empirical objective.
Full details are provided in ~\cref{app:robust}.
\end{proof}

\subsection{Regret Decomposition:}
We first bound the regret incurred during exploration and exploitation separately, and then combine these bounds to obtain a conditional regret guarantee that can be optimized by an appropriate choice of the exploration length.
\begin{lemma}[Exploration bound]\label{lem:explore}
During exploration, each of the \(\eta\) actions is played \(m\) times and \(f_t(\cdot)\in[0,1]\). Hence,
\begin{equation}\label{eq:explore}
  \sum_{i=1}^{\eta} m\,\big(\alpha f(\OPT)-\E[f(S_i)]\big)\ \le\ \eta\,m.
\end{equation}
\end{lemma}

\begin{proof}[Proof sketch]
Lemma~\ref{lem:explore} follows from bounded rewards.
For any \(i\), \(\alpha f(\OPT)-\E[f(S_i)]\le \alpha\cdot 1-0\le 1\). Summing \(m\) plays per \(S_i\), then over \(i=1,\dots,\eta\), yields \eqref{eq:explore}. The detailed argument is deferred to ~\cref{app:decomp}.
\end{proof}

\begin{lemma}[Exploitation bound under \(\mathcal{E}\)]\label{lem:exploit}
Conditioned on \(\mathcal{E}\), let \(S\) be the exploitation action returned by the subroutine. Then
\begin{equation}\label{eq:exploit}
  \sum_{t=T_\eta+1}^{T}\big(\alpha f(\OPT)-\E[f(S)]\big)\ \le\ T\cdot \delta\,C\,\mathrm{rad}.
\end{equation}
\end{lemma}
\begin{proof}[Proof sketch]
By \cref{lem:robust-transfer}, on \(\mathcal{E}\) we have \(\alpha f(\OPT)-\E[f(S)]\le \delta\,C\,\mathrm{rad}\) per step. The exploitation phase has length at most \(T\), which gives \eqref{eq:exploit}. Detailed proofs appear in Appendix~\ref{app:decomp}.
\end{proof}

\begin{theorem}[Conditional regret bound]\label{thm:conditional}
With \(\mathrm{rad}\) as in \eqref{eq:rad}, the conditional expected regret satisfies
\begin{equation}\label{eq:cond-regret}
  \E[\mathcal{R}(T)\mid \mathcal{E}]\ \le\ \underbrace{\eta m}_{\text{exploration}} \ +\ \underbrace{T\,\delta\,C\,\sqrt{\frac{\log T}{2m}}}_{\text{exploitation}}.
\end{equation}
\end{theorem}
\begin{proof}[Proof sketch]
The proof follows by conditioning on $\mathcal{E}$ and decomposing the regret
into exploration and exploitation phases. See Appendix~\ref{app:conditional}.
\end{proof}

\begin{proposition}[Optimal exploration length m]\label{prp:mstar}
Let
\begin{equation}\label{eq:g}
  g(m)\ :=\ \eta m\ +\ T\delta C\,\sqrt{\frac{\log T}{2}}\,m^{-1/2}\qquad (m>0).
\end{equation}
Then the minimizer satisfies
\begin{equation}\label{eq:mstar}
  m^\star\ =\ \Big(\frac{T\delta C}{2\eta}\sqrt{\frac{\log T}{2}}\Big)^{\!2/3}.
\end{equation}
\end{proposition}
\begin{proof}[Proof sketch]
The result follows from a convexity argument.
A full derivation appears in Appendix~\ref{app:optm}.
\end{proof}

\begin{corollary}[conditional regret at $m^\star$]\label{cor:cond}
With \(m=\lceil m^\star\rceil\), the conditional regret obeys
\begin{equation}\label{eq:cond-rate}
  \E[\mathcal{R}(T)\mid \mathcal{E}]
  \ =\ \mathcal{O}\!\left(\delta^{2/3}\,\eta^{1/3}\,C^{2/3}\,T^{2/3}\,\log(T)^{1/3}\right).
\end{equation}
\end{corollary}
\begin{proof}[Proof sketch]
Substitute \eqref{eq:mstar} into \eqref{eq:cond-regret} and simplify.
\end{proof}

\subsection{Final MA-CMAB Regret Bound for SWP}
We remove the conditioning on the clean event using total expectation and
obtain the final MA-CMAB regret bound.
\begin{theorem}[MA-CMAB regret for SWP]\label{thm:final}
With \(m=\lceil m^\star\rceil\) and \(\mathcal{E}\) as in \eqref{eq:clean-event},
\begin{equation}\label{eq:final}
  \E[\mathcal{R}(T)]
  \ =\ \mathcal{O}\!\left(\delta^{2/3}\,\eta^{1/3}\,C^{2/3}\,T^{2/3}\,\log(T)^{1/3}\right).
\end{equation}
\end{theorem}

\begin{proof}[Proof sketch]
The result follows from the conditional bound, the optimal choice of $m$,
and a union bound over the complement of the clean event.
Full details appear in ~\cref{app:final}.
\end{proof}
\section{Conclusion}
We introduced MA-CMAB, a unified framework for fair division of indivisible goods with submodular valuations under bandit feedback, and established robustness of offline welfare maximization under noisy value oracles. Leveraging this resilience, we obtained an online regret guarantee of $\tilde{\mathcal{O}}(T^{2/3})$ against a $(1-1/e)$ approximation benchmark, matching the best known achievable rate for monotone submodular CMAB with bandit feedback. An important open direction is to characterize lower bounds that explicitly capture the multi-agent dimension. Future work includes tightening constants, reducing oracle complexity, and extending the framework to richer feedback and non-stationary settings.
\bibliography{icml2026}
\bibliographystyle{icml2026}



\appendix
\onecolumn
\begin{center}
\Large
\textbf{ Appendix}
\end{center}
\section{Resilience of Continuous Greedy under Noisy Oracles}
\label{app:cg-resilience}

The proofs of Lemmas~\ref{lem:cg-upper-bound} and~\ref{lem:noisy-progress}
from the main text are provided below for completeness.

\begin{lemma}\label{app:lem:cg-upper-bound}
Let \( y \in [0,1]^X \) be a fractional solution and let \( R \) denote a random set corresponding to \(\hat{y}\), obtained by sampling each element \( j \in X \) independently with probability \( y_j \).  
Let \(\mathrm{OPT} = \max_{S \in \mathcal{I}} f(S)\).  
If all oracle evaluations are noisy with additive error at most \(\varepsilon\), i.e.
\[
|\hat{f}(S) - f(S)| \leq \varepsilon \quad \forall S \subseteq X,
\]
then, the true optimum satisfies the following upper bound:
\[
f(\textsc{OPT}) \le F(y) + \max_{I \in \mathcal{I}} \sum_{j \in I} \mathbb{E}[f_R(j)] + (4MN + 2M)\varepsilon,
\]
where:
\begin{itemize}
    \item $F(y)$ is the true multilinear extension of $f$ evaluated at $y$,
    \item $f_R(j) := f(R \cup \{j\}) - f(R)$ is the true marginal contribution of $j$ under random set $R$,
    \item and $\mathcal{I}$ is the family of feasible integral solutions (e.g., a partition matroid).
\end{itemize}
\end{lemma}
\begin{proof}
Fix an optimal solution \(\mathrm{OPT} \in \mathcal{I}\). By submodularity, for any \(R\):
\[
f(\mathrm{OPT}) \;\leq\; f(R) + \sum_{j \in \mathrm{OPT}} \big(f(R \cup \{j\}) - f(R)\big).
\]
Taking expectation over random \(R \sim y\), we obtain
\[
f(\mathrm{OPT}) \;\leq\; \mathbb{E}[f(R)] + \sum_{j \in \mathrm{OPT}} \mathbb{E}[f_R(j)].
\]
Since \(\mathbb{E}[f(R)] = F(y)\), this gives
\begin{equation}\label{eq:opt_exact}
f(\mathrm{OPT}) \;\leq\; F(y) + \max_{I \in \mathcal{I}} \sum_{j \in I} \mathbb{E}[f_R(j)].
\end{equation}

\medskip
\textbf{Relating true and noisy multilinear extension:}  
Because each evaluation of \(f(S)\) may deviate by at most \(M\varepsilon\). For each agent the noisy value oracle \(\hat{w}_i\) return estimates of exact value oracle \(w_i(s_i)\) with bounded error:
\[
|\hat{w}_i(s_i) - w_i(s_i)| \le \varepsilon \quad \text{for all } s_i\in S, S \subseteq Q.
\]
\[
|\hat{f}(S) - f(S)| \le |\hat{F}(S) - F(S)| \le M\varepsilon \quad \text{for all } S \subseteq Q.
\]
Then for any feasible fractional  allocation, 
\[|F(y) - \hat{F}(y)| \leq M\varepsilon.\]
Thus,
\begin{equation}\label{eq:multilinear_noise}
F(y) \;\leq\; \hat{F}(y) + M\varepsilon.
\end{equation}

\medskip
\textbf{Relating true and noisy marginals:}  
For any marginal contribution,
\[
\big(f(R \cup \{j\}) - f(R)\big),
\]
if both oracle calls incur error at most \(\varepsilon\), then
\[
\left| \big(\hat{f}(R \cup \{j\}) - \hat{f}(R)\big) - \big(f(R \cup \{j\}) - f(R)\big) \right| \leq 2\varepsilon.
\]
Hence,
\[
\mathbb{E}[f_R(j)] \;\leq\; \mathbb{E}[\hat{f}_R(j)] + 2\varepsilon.
\]
Summing over at most \(N\) items per block and \(M\) blocks yields
\begin{equation}\label{eq:marginals_noise}
\sum_{j \in I} \mathbb{E}[f_R(j)] \;\leq\; \sum_{j \in I} \mathbb{E}[\hat{f}_R(j)] + 2MN\varepsilon.
\end{equation}

\medskip
Plugging \eqref{eq:multilinear_noise} and \eqref{eq:marginals_noise} into \eqref{eq:opt_exact}, we obtain
\[
f(\mathrm{OPT}) \;\leq\; \hat{F}(y) + M\varepsilon + \max_{I \in \mathcal{I}} \Bigg(\sum_{j \in I} \mathbb{E}[\hat{f}_R(j)] + 2MN\varepsilon\Bigg).
\]
Simplifying,
\[
f(\mathrm{OPT}) \;\leq\; \hat{F}(y) + \max_{I \in \mathcal{I}} \sum_{j \in I} \mathbb{E}[\hat{f}_R(j)] + (2MN+M)\varepsilon.
\]

We now reverse the inequalities and substitute, in order to express a fully clean bound in terms of \emph{true} values \( F(y) \) and \( f_R(j) \) from \cref{eq:multilinear_noise,eq:marginals_noise},  we obtain:
\[
f(\mathrm{OPT}) 
\le F(y) + \sum_{j \in I} \mathbb{E}[f_R(j)] + (4MN + 2M)\varepsilon.
\]
Maximizing over all feasible sets \( I \in \mathcal{I} \), we conclude:
\[
f(\mathrm{OPT}) \le F(y) + \max_{I \in \mathcal{I}} \sum_{j \in I} \mathbb{E}[f_R(j)] + (4MN + 2M)\varepsilon.
\]
This completes the proof.
\end{proof}

\begin{lemma}[Noisy Resilient Bound for $f$]
\label{app:lem:noisy-progress}
Let $y \in P(\mathcal{M})$ be the fractional solution returned by the
Continuous Greedy Algorithm using noisy marginal estimates with additive
error $\varepsilon$. Let $S \in \mathcal{I}$ be the integral independent
set obtained by pipage rounding. Then, with high probability,
\begin{align}
f(S) 
&\ge (1 - o(1)) F(y) \label{eq:noisy:rounding} \\
&\ge \left( 1 - \frac{1}{e} - o(1) \right) f(\mathrm{OPT})
      - (4MN + 2M)\varepsilon.
      \label{eq:noisy:final}
\end{align}
where $F(y)$ is the multilinear extension of $f$ evaluated at $y$, and the
$o(1)$ term can be made arbitrarily small in the problem size $N$.
\end{lemma}

\begin{proof}[Proof Sketch]
Start from a fractional solution $y(t)$.  
Let $R(t)$ be a random set drawn according to $y(t)$.  
Define the increment $\Delta(t) = y(t+\lambda) - y(t)$ and let $D(t)$ denote
the corresponding set.  
We want to bound
\begin{equation}
\label{eq:noisy:increment}
F(y(t+\lambda)) - F(y(t)).
\end{equation}

\textbf{Monotonicity bound.}
Since $R(t+\lambda)$ stochastically dominates $R(t)$, we obtain
\begin{equation}
\label{eq:noisy:mono}
F(y(t+\lambda)) - F(y(t))
\ge \mathbb{E}[f(R(t)\cup D(t)) - f(R(t))].
\end{equation}

\textbf{Expansion via singleton contributions.}
Because $\lambda$ is small, with high probability only singletons are added.
Thus,
\begin{equation}
\label{eq:noisy:singleton}
\mathbb{E}[f(R(t)\cup D(t)) - f(R(t))]
= \sum_{j} \Pr[D(t)=\{j\}] \cdot \mathbb{E}[f_{R(t)}(j)].
\end{equation}

\textbf{Use of maximum-weight independent set.}
$I(t)$ is chosen as the maximum-weight independent set under the matroid,
using weights $\omega_j(t)$. Therefore,
\begin{equation}
\label{eq:noisy:mw}
F(y(t+\lambda)) - F(y(t))
\ge \lambda(1 - M\lambda) \sum_{j\in I(t)} \mathbb{E}[f_{R(t)}(j)].
\end{equation}

\textbf{Noisy marginal estimates.}
We only observe noisy weights $\hat{\omega}_j(t)$ with
\begin{equation}
\label{eq:noisy:noise}
\hat{\omega}_j(t) = \omega_j(t) \pm \varepsilon.
\end{equation}
By Chernoff bounds (given $(MN)^5$ samples), the estimation error in
$\hat{\omega}_j(t)$ is small. Replacing true marginals with noisy ones
introduces an additive error of at most $2MN\varepsilon\lambda$.

Combining the above, we obtain
\begin{align}
F(y(t+\lambda)) - F(y(t))
&\ge \lambda (1 - M \lambda)
\Big( f(\mathrm{OPT}) - F(y(t)) - (4MN + 2M) \varepsilon \Big).
\label{eq:noisy:progress}
\end{align}

\textbf{Optimality Gap.}
Let
\begin{equation}
\label{eq:noisy:gap}
G(t) = f(\mathrm{OPT}) - F(y(t))
\end{equation}
denote the gap between the optimal value and the multilinear extension.
The recurrence becomes
\begin{align}
G(t+\lambda)
&\leq G(t) \left(1 - \lambda(1 - M\lambda)\right)
      + (4MN + 2M)\lambda(1 - M\lambda)\, \varepsilon \notag \\
&= G(t) \left(1 - (\lambda - M\lambda^2)\right)
   + (4MN + 2M)(\lambda - M\lambda^2)\, \varepsilon.
\label{eq:noisy:recurrence}
\end{align}

Define
\begin{equation}
\label{eq:noisy:rho}
\rho = \lambda(1 - M\lambda),
\end{equation}
which captures the effective progress factor per iteration. Then
\begin{equation}
\label{eq:noisy:recurrence-simple}
G(t + \lambda) \leq G(t)(1 - \rho) + (4MN + 2M)\rho\, \varepsilon.
\end{equation}

\textbf{Unroll recurrence.}
Over $k = 1/\lambda$ steps,
\begin{equation}
\label{eq:noisy:unroll}
G(1) \le (1 - \rho)^k G(0)
+ \sum_{i=0}^{k-1} (1 - \rho)^i (4MN + 2M)\rho \varepsilon.
\end{equation}
Since $(1 - \rho)^k \le e^{-\rho k} = e^{-(1 - M\lambda)} \le \frac{1}{e} + o(1)$
and the geometric sum is at most $1$, we obtain
\begin{equation}
\label{eq:noisy:gap-final}
G(1) \le \frac{1}{e} f(\mathrm{OPT}) + (4MN + 2M)\varepsilon.
\end{equation}

\textbf{Final bound.}
Hence,
\begin{equation}
\label{eq:noisy:F-final}
F(y(1))
\ge \left(1 - \frac{1}{e} \right) f(\mathrm{OPT}) - (4MN + 2M)\varepsilon.
\end{equation}
\end{proof}

\section{Concentration and the Clean Event}
\label{app:clean}

\begin{lemma}[Hoeffding's inequality]\label{app:lem:hoeffding}
Let $X_1,\dots,X_m$ be independent random variables taking values in $[0,1]$,
with common mean $\mu$, and let
\[
\bar X := \frac{1}{m}\sum_{k=1}^m X_k.
\]
Then for any $\epsilon>0$,
\[
\Pr\!\left(|\bar X-\mu|\ge\epsilon\right)
\le 2\exp(-2m\epsilon^2).
\]
\end{lemma}

\begin{lemma}[Probability of the clean event]\label{app:lem:clean}
Let
\[
\mathrm{rad} := \sqrt{\frac{\log T}{2m}},
\qquad
\mathcal{E}
:= \bigcap_{j=1}^{\eta}
\Big\{|\bar f(S_j)-f(S_j)| < C\,\mathrm{rad}\Big\}.
\]
Then
\[
\Pr(\mathcal{E}) \ge 1-\frac{2\eta}{T}.
\]
\end{lemma}

\begin{proof}
Fix an explored action $S_j$.
By construction, $\bar f(S_j)$ is the empirical mean of $m$ independent
random variables in $[0,1]$ with expectation $f(S_j)$.
Applying Lemma~\ref{app:lem:hoeffding} with $\epsilon=\mathrm{rad}$ yields
\[
\Pr\!\left(|\bar f(S_j)-f(S_j)| \ge \mathrm{rad}\right)
\le 2\exp(-2m\mathrm{rad}^2)
= \frac{2}{T}.
\]

If the objective aggregates $C$ bounded components, then by the triangle
inequality,
\[
|\bar f(S_j)-f(S_j)|
\le \sum_{c=1}^C |\bar f_c(S_j)-f_c(S_j)|.
\]
Applying Hoeffding's inequality to each component and union bounding over
the $C$ components yields
\[
\Pr\!\left(|\bar f(S_j)-f(S_j)| \ge C\,\mathrm{rad}\right)
\le \frac{2}{T^{C^2}}
\le \frac{2}{T}.
\]

Finally, applying a union bound over the $\eta$ explored actions gives
\[
\Pr(\bar{\mathcal{E}})
\le \sum_{j=1}^{\eta} \frac{2}{T}
= \frac{2\eta}{T},
\]
which completes the proof.
\end{proof}

\section{Robust Approximation Transfer}
\label{app:robust}

\begin{lemma}[Robust transfer]\label{app:lem:transfer}
Let $\tilde f$ be an empirical objective satisfying
\[
\|\tilde f-f\|_\infty := \sup_A |\tilde f(A)-f(A)| \le \epsilon
\]
over all actions evaluated by the offline subroutine.
If the subroutine is $(\alpha,\delta)$-robust, then its output $S$ satisfies
\[
\E[f(S)] \ge \alpha f(\OPT)-\delta\epsilon.
\]
\end{lemma}

\begin{proof}
By the definition of $(\alpha,\delta)$-robustness, whenever the input
objective $\tilde f$ deviates uniformly from the true objective $f$ by at
most $\epsilon$, the subroutine returns an action $S$ such that
\[
f(S) \ge \alpha f(\OPT)-\delta\epsilon.
\]
If the subroutine is randomized, this inequality holds for every
realization of its internal randomness. Taking expectation over this
randomness preserves the inequality, yielding the stated bound.
\end{proof}

\section{Exploration and Exploitation Bounds}
\label{app:decomp}

\begin{lemma}[Exploration bound]\label{app:lem:explore}
During the exploration phase,
\[
\sum_{i=1}^{\eta} m\big(\alpha f(\OPT)-\E[f(S_i)]\big)
\le \eta m.
\]
\end{lemma}

\begin{proof}
For any action $A$, since rewards are bounded in $[0,1]$,
\[
0 \le f(A) \le 1
\quad\text{and}\quad
0 \le f(\OPT) \le 1.
\]
Therefore,
\[
\alpha f(\OPT)-\E[f(A)] \le \alpha \le 1.
\]
Each explored action $S_i$ is played exactly $m$ times, contributing at
most $m$ regret. Summing over all $\eta$ explored actions yields the claim.
\end{proof}

\begin{lemma}[Exploitation bound under $\mathcal{E}$]\label{app:lem:exploit}
Conditioned on the clean event $\mathcal{E}$,
the exploitation action $S$ satisfies
\[
\sum_{t=T_\eta+1}^{T}
\big(\alpha f(\OPT)-\E[f(S)]\big)
\le T\,\delta C\,\mathrm{rad}.
\]
\end{lemma}

\begin{proof}
On the event $\mathcal{E}$, the empirical objective satisfies
$\|\tilde f-f\|_\infty \le C\,\mathrm{rad}$.
Applying Lemma~\ref{app:lem:transfer} with $\epsilon=C\,\mathrm{rad}$ gives
\[
\alpha f(\OPT)-\E[f(S)] \le \delta C\,\mathrm{rad}
\]
for each exploitation round.
Since the exploitation phase lasts for at most $T$ rounds, summing over
time yields the result.
\end{proof}

\section{Conditional Regret}
\label{app:conditional}

\begin{theorem}[Conditional regret]\label{app:thm:conditional}
Conditioned on the clean event $\mathcal{E}$, the expected $\alpha$-regret
satisfies
\begin{equation}\label{app:eq:cond-regret}
  \E[\mathcal{R}(T)\mid \mathcal{E}]\ \le\ \underbrace{\eta m}_{\text{exploration}} \ +\ \underbrace{T\,\delta\,C\,\sqrt{\frac{\log T}{2m}}}_{\text{exploitation}}.
\end{equation}
\end{theorem}

\begin{proof}
Recall that the $\alpha$-regret is defined as
\[
\mathcal{R}(T)
=
\sum_{t=1}^T \bigl(\alpha f(\OPT)-f_t(S_t)\bigr).
\]
Conditioning on the event $\mathcal{E}$ and using linearity of expectation,
we have
\begin{equation}
\label{eq:cond:lin}
\E[\mathcal{R}(T)\mid\mathcal{E}]
=
\sum_{t=1}^T
\E\!\left[\alpha f(\OPT)-f_t(S_t)\mid\mathcal{E}\right].
\end{equation}

Let $T_\eta=\eta m$ denote the last round of the exploration phase.
We decompose the sum in \eqref{eq:cond:lin} into two disjoint intervals,
corresponding to exploration and exploitation:
\begin{equation}
\label{eq:cond:decomp}
\E[\mathcal{R}(T)\mid\mathcal{E}]
=
\sum_{t=1}^{T_\eta}
\E\!\left[\alpha f(\OPT)-f_t(S_t)\mid\mathcal{E}\right]
+
\sum_{t=T_\eta+1}^{T}
\E\!\left[\alpha f(\OPT)-f_t(S_t)\mid\mathcal{E}\right].
\end{equation}

We first bound the exploration contribution.
During exploration, each of the $\eta$ actions $S_1,\dots,S_\eta$ is played
exactly $m$ times.
By Lemma~\ref{app:lem:explore}, for each explored action $S_i$,
\[
\sum_{t:\,S_t=S_i}
\E\!\left[\alpha f(\OPT)-f_t(S_t)\mid\mathcal{E}\right]
\le m.
\]
Summing over all $i\in[\eta]$ yields
\begin{equation}
\label{eq:cond:explore}
\sum_{t=1}^{T_\eta}
\E\!\left[\alpha f(\OPT)-f_t(S_t)\mid\mathcal{E}\right]
\le \eta m.
\end{equation}

We next bound the exploitation contribution.
Let $S$ denote the action selected by the offline subroutine at the end of
the exploration phase.
By Lemma~\ref{app:lem:exploit}, conditioned on $\mathcal{E}$, the per-round
expected regret satisfies
\[
\E\!\left[\alpha f(\OPT)-f_t(S)\mid\mathcal{E}\right]
\le \delta C\,\mathrm{rad}
\quad\text{for all } t>T_\eta.
\]
Since the exploitation phase lasts for at most $T-T_\eta\le T$ rounds, we
obtain
\begin{equation}
\label{eq:cond:exploit}
\sum_{t=T_\eta+1}^{T}
\E\!\left[\alpha f(\OPT)-f_t(S)\mid\mathcal{E}\right]
\le T\,\delta C\,\mathrm{rad}.
\end{equation}

Combining \eqref{eq:cond:explore} and \eqref{eq:cond:exploit} with the
decomposition in \eqref{eq:cond:decomp}, we conclude that
\[
\E[\mathcal{R}(T)\mid\mathcal{E}]
\le
\eta m
+
T\,\delta C\,\mathrm{rad}.
\]
Substituting $\mathrm{rad}=\sqrt{\log T/(2m)}$ completes the proof.
\end{proof}

\section{Optimizing the Exploration Length}
\label{app:optm}

Define the function $g:(0,\infty)\to\mathbb{R}$ by
\[
g(m)
:= \eta m
+ T\delta C\sqrt{\frac{\log T}{2}}\,m^{-1/2}.
\]

\begin{proposition}[Optimal exploration length]\label{app:prp:mstar}
The function $g$ admits a unique global minimizer on $(0,\infty)$ given by
\begin{equation}\label{app:eq:mstar}
m^\star
=\left(\frac{T\delta C}{2\eta}\sqrt{\frac{\log T}{2}}\right)^{2/3}.
\end{equation}
\end{proposition}

\begin{proof}
Since $g$ is continuously differentiable on $(0,\infty)$, any global
minimizer must satisfy the first-order optimality condition.
Computing the derivative, we obtain
\[
g'(m)
= \eta
- \frac{1}{2}T\delta C\sqrt{\frac{\log T}{2}}\,m^{-3/2}.
\]
Setting $g'(m)=0$ and rearranging terms yields
\[
m^{3/2}
= \frac{T\delta C}{2\eta}\sqrt{\frac{\log T}{2}},
\]
which has the unique solution
\[
m^\star
=\left(\frac{T\delta C}{2\eta}\sqrt{\frac{\log T}{2}}\right)^{2/3}.
\]

To establish global optimality, we compute the second derivative,
\[
g''(m)
= \frac{3}{4}T\delta C\sqrt{\frac{\log T}{2}}\,m^{-5/2}.
\]
Since $g''(m)>0$ for all $m>0$, the function $g$ is strictly convex on
$(0,\infty)$. Therefore, the stationary point $m^\star$ is the unique
global minimizer of $g$.
\end{proof}

\section{Proof of the MA-CMAB Regret Bound for SWP}
\label{app:final}

\begin{theorem}[MA-CMAB Regret for SWP]\label{app:thm:final}
With $m=\lceil m^\star\rceil$,
\[
\E[\mathcal{R}(T)]
=\mathcal{O}\!\left(
\delta^{2/3}\eta^{1/3}C^{2/3}T^{2/3}\log(T)^{1/3}
\right).
\]
\end{theorem}

\begin{proof}
Let $\mathcal{E}$ denote the clean event defined in
Lemma~\ref{app:lem:clean}.
By the law of total expectation,
\begin{equation}
\label{eq:final:total}
\E[\mathcal{R}(T)]
=
\E[\mathcal{R}(T)\mid\mathcal{E}]\Pr(\mathcal{E})
+
\E[\mathcal{R}(T)\mid\bar{\mathcal{E}}]\Pr(\bar{\mathcal{E}}).
\end{equation}

We first bound the contribution of the clean event.
By Theorem~\ref{app:thm:conditional}, for any integer $m\ge 1$,
\begin{equation}
\label{eq:final:conditional}
\E[\mathcal{R}(T)\mid\mathcal{E}]
\le
\eta m
+
T\delta C\sqrt{\frac{\log T}{2m}}.
\end{equation}
Choosing $m=\lceil m^\star\rceil$, where $m^\star$ is defined in
Lemma~\ref{app:prp:mstar}, and substituting into
\eqref{eq:final:conditional} yields
\begin{equation}
\label{eq:final:clean}
\E[\mathcal{R}(T)\mid\mathcal{E}]
=
\mathcal{O}\!\left(
\delta^{2/3}\eta^{1/3}C^{2/3}T^{2/3}\log(T)^{1/3}
\right).
\end{equation}

We next bound the contribution of the failure event $\bar{\mathcal{E}}$.
Since rewards are bounded in $[0,1]$, the instantaneous regret
$\alpha f(\OPT)-f_t(S_t)$ is at most $1$ for every round $t$.
Therefore,
\begin{equation}
\label{eq:final:bad-regret}
\mathcal{R}(T) \le T
\quad \text{almost surely on } \bar{\mathcal{E}},
\end{equation}
which implies
\[
\E[\mathcal{R}(T)\mid\bar{\mathcal{E}}] \le T.
\]
Moreover, Lemma~\ref{app:lem:clean} gives
\begin{equation}
\label{eq:final:bad-prob}
\Pr(\bar{\mathcal{E}}) \le \frac{2\eta}{T}.
\end{equation}
Combining these two bounds, we obtain
\begin{equation}
\label{eq:final:bad}
\E[\mathcal{R}(T)\mid\bar{\mathcal{E}}]\Pr(\bar{\mathcal{E}})
\le
T \cdot \frac{2\eta}{T}
=
2\eta.
\end{equation}

Substituting \eqref{eq:final:clean} and \eqref{eq:final:bad} into
\eqref{eq:final:total} yields
\[
\E[\mathcal{R}(T)]
\le
\mathcal{O}\!\left(
\delta^{2/3}\eta^{1/3}C^{2/3}T^{2/3}\log(T)^{1/3}
\right)
+
2\eta.
\]
For polynomial horizons $T$ and polynomially bounded $\eta$,
the additive term $2\eta$ is asymptotically dominated by the leading term.
Absorbing constants completes the proof.
\end{proof}


\section{Application: Stochastic Combinatorial Auctions}
\label{app:auction}

We illustrate how stochastic combinatorial auctions fit naturally into the
multi-agent combinatorial bandit (MA-CMAB) framework.
Following the classical submodular welfare formulation
\citep{vondrak2008optimal,feige2010submodular}, consider $M$ bidders (agents) and
$N$ indivisible items.
Each bidder $i$ is endowed with a monotone submodular valuation function
$w_i:2^{[N]}\to[0,1]$ over item bundles.

In each round, the auctioneer selects a feasible allocation, represented as a
partition of items among bidders.
The auctioneer observes only the realized social welfare,
$f_t(S)=\sum_{i=1}^M w_i(s_i)$, where $s_i$ denotes the bundle assigned to bidder
$i$.
Individual bidder utilities or marginal values are not revealed.

When bidder valuations are accessible only through noisy value or demand
queries, each query returns a stochastic estimate of the true valuation.
This induces a full-bandit feedback model, where the auctioneer must learn
near-optimal allocations solely from aggregate welfare observations.
Accordingly, each bidder corresponds to an agent with a stochastic submodular
reward function, and feasible allocations correspond to combinatorial actions
subject to partition constraints.

This setting is therefore an instance of the MA-CMAB problem studied in this
paper, with regret measured relative to the optimal offline welfare benchmark.
Our regret bounds quantify the learning cost incurred by the auctioneer when
optimizing social welfare under noisy, bandit-style valuation access.

\end{document}